\documentclass[12pt,oneside]{article}
\usepackage[utf8]{inputenc} 
\usepackage{amsthm,verbatim,mathtools}
\usepackage{amssymb,latexsym}
\usepackage{amsmath}
\usepackage{graphicx}
\usepackage{color}
\usepackage{enumerate}
\usepackage{cancel}
\usepackage{soul}
\usepackage{ulem}
\newtheorem{theorem}{Theorem}[section]
\newtheorem{definition}[theorem]{Definition}
\newtheorem{remark}[theorem]{Remark}
\newtheorem{lemma}[theorem]{Lemma}

\newcommand{\cH}{\mathcal H}

\newcommand{\fq}{{\mathbb{F}_q}}
\newcommand{\fqs}{{\mathbb{F}_{q^2}}}

\newcommand{\cf}{\mathcal{F}_{q, r}}
\begin{document}
\title{Construction of sequences with high nonlinear complexity from a generalization of the Hermitian function field}

%\author{Alonso S. Castellanos, Luciane Quoos and Guilherme Tizziotti}
\date{}

\author{A.S. Castellanos\thanks{Alonso S. Castellanos is with Universidade Federal de Uberl\^andia, Campus Santa M\^onica, Av. Jo\~ao Naves de Avila 2121, Uberl\^andia-MG, CEP 38.408-100, Brazil, {\em email:} alonso.castellanos@ufu.br, {\em orcid:} 0000-0003-2862-5339.},\,
L. Quoos\thanks{Luciane Quoos is with Universidade Federal do Rio de Janeiro, Centro de Tecnologia - Bloco C, Cidade Universitária - Av. Athos da Silveira Ramos, 149 - Ilha do Fundão, CEP 21.941-909 - Brazil, {\em email: } luciane@im.ufrj.br, {\em orcid:} 0000-0001-6310-0859.} \, and \,
G. Tizziotti\thanks{Guilherme Tizziotti is with the Universidade Federal de Uberl\^andia, Campus Santa M\^onica, Av. Jo\~ao Naves de Avila 2121, Uberl\^andia-MG, CEP 38.408-100, Brazil, {\em email: } guilhermect@ufu.br, {\em orcid:} 0000-0003-1026-0546.}
}
%\affil{}
%\thanks{The second author thanks Conselho Nacional de Desenvolvimento Científico e Tecnológico-CNPq, Fundação de Amparo à Pesquisa do Estado do Rio de Janeiro - FAPERJ, and Coordenação de Aperfeiçoamento de Pessoal de Nível Superior-CAPES for the support of this research.}

\maketitle

\begin{abstract}
For $r \geq 1$ an odd integer, we provide a sequence from the function field $\cf$ of the maximal curve over $\mathbb{F}_{q^{2r}}$ defined by the affine equation $y^q+y=x^{q^r + 1}$. This sequence has high nonlinear complexity, and this fact comes from the existence of a rational function on $\cf$ with pole divisor of small degree, and support in certain $q$ rational places. 

%In particular, for $r=1$, we improve the lower bounds on the $k$th-order nonlinear complexity obtained by H. Niederreiter and C. Xing \cite{NX};  and O. Geil, F. \"Ozbudak and D. Ruano \cite{GOR} for the Hermitian function field.

\end{abstract}

{\small{\bf MSC codes (2020):} 14G50; 94A55}

\section{Introduction}

In recent years the study of sequences over finite fields has increased due to its applications in cryptography and pseudorandom number generation, see e.g. \cite{GOR, YXY, MW} and \cite{NX}. In cryptography, to assess the suitability of a pseudorandom sequence, the complexity-theoretic and statistical requirements must be tested. In practice, both categories of tests - complexity-theoretic and statistical - should be carried out since these two categories are in a sense independent. For details on the testing of pseudorandom sequences in a cryptographic context, we refer the reader to \cite{Rueppel}. In this paper, we are concerned about the complexity-theoretic analysis of pseudorandom sequences. 

One of the requirements of such sequences is that it should be very hard to replicate the entire sequence from the knowledge of a part of it, that is, its complexity should be large. Many different complexity measures are available in the literature, the most usual being the {\it linear complexity}. 

Let $\fq$ be the finite field with $q$ elements, where $q$ is the power of a prime $p$.
The linear complexity $L(\bf{t_n})$ of a sequence $\mathbf{t}_n=(t_1, \ldots , t_n)$ over $\fq$  is the length $m $ of a shortest linear recurrence relation
$$
t_{j+m} = c_{m-1}t_{j+m-1} +\cdots +c_0t_j, \quad 0 \leq j \leq n-m-1,
$$ 
where $c_j \in \mathbb{F}_q$ for all $j=0,1,\ldots, m-1$. We refer to the recent handbook \cite{MW} for a concise survey on linear complexity of sequences. 

In the past years various types of sequences  with a large linear
complexity have been constructed from function fields , see $\cite{ XKD}$ and \cite{XL}. In \cite{NX}, Niederreiter and Xing investigated the complexity measure of certain sequences related to feedback shift registers with feedback functions of higher algebraic degree, the so-called \textit{nonlinear complexity}.
More precisely, for a non zero sequence $\mathbf{t}_n$ of length $n \geq 1$ over a finite field $\mathbb{F}_q$ and $k \in \mathbb{N}$, the nonlinear complexity  $N^{(k)}(\mathbf{t}_n)$ is defined as the smallest $m \in \mathbb{N}$ for which there exists a polynomial $f \in \mathbb{F}_q[x_1, \ldots , x_m]$ of degree at most $k$ in each variable such that
$$
t_{i+m} = f(t_i, t_{i+1}, \ldots , t_{i+m-1}), \mbox{ for } 1 \leq i \leq n-m.
$$
If we allow the total degree of the polynomial $f$ to be at most $k$ a similar nonlinear complexity $L^{(k)}(\mathbf{t}_n)$ can be found, see Definition $\ref{order nonlinear 2}$. 

In this paper, we focus on constructing sequences with large nonlinear complexity. Recently, constructions of sequences from function fields with high nonlinear complexity  arose in the work of Niederreiter and Xing \cite{NX}; Luo, Xing, and You \cite{YXY};  and Geil, \"Ozbudak and Ruano \cite{GOR}. In \cite{YXY}, the authors construct sequences over $\fq$ using the rational and cyclotomic function fields. In \cite{NX, GOR}, the authors work over the Hermitian function field $\cH$  to construct sequences over $\fqs$ of length $(q-1)(q^2-1)$ with large nonlinear complexity. Both constructions rely on the action of a  large cyclic subgroup of the automorphism group $Aut(\cH)$ of the function field $\cH$ on the $\fqs$-rational places of $\cH$.
%and, for two fixed rational places $P_1$ and $ P_2$ on the Hermitian function field, on the existence of certain functions with poles only in $P_1$ and $ P_2$. A meticulous study on the Weierstrass semigroup $H(P_1, P_2)$ in two places
%$$
%H(P_1, P_2)=\{ (a, b) \in \mathbb{N}^2 \mid \exists \ f \in \cH \text{ with pole divisor } (f)_\infty=aP_1+bP_2 \},
%$$ 
%allows to construct sequences $\mathbf{t}_n$ with improvements on the nonlinear complexities $N^{(k)}(\mathbf{t}_n)$ and $L^{(k)}(\mathbf{t}_n)$. 
The results in $\cite[\mbox{Theorems } 1, 2]{GOR}$ improve the
bounds on the nonlinear complexities given in $\cite[\mbox{Theorems } 3, 4]{NX}$.

%: the lower bound for $N^{(k)}(\mathbf{t}_n)$ for all the parameters and, for some range of the parameters, also improve the lower bound on $L^{(k)}(\mathbf{t}_n)$.  

In this work, we consider a generalization of the Hermitian curve, the maximal curve defined over $\mathbb{F}_{q^{2r}}$ by the affine equation $y^q + y = x^{q^r + 1}$, where $r\geq 1$ is an odd integer. Denote by  $\mathcal{F}_{q, r}$ its function field and by $P_\infty$ the only pole of $x$ and $y$.

We propose a new sequence, inspired by the sequences presented in $\cite{GOR}$ and $\cite{NX}$, from the function field $\mathcal{F}_{q, r}$ over $\mathbb{F}_{q^{2r}}$ of length $q(q^{2r}-q^{r-1}-q^{r-2} - \ldots - q - 2)$. The constructed sequence relies on the existence of functions in $\mathcal{F}_{q, r}$ with pole divisor in $q$  fixed rational places $P_\infty, Q_2, \ldots, Q_{q}$ on $\mathcal{F}_{q, r}$. These functions are determined using a characterization of the Weierstrass semigroup 
$$
H(P_\infty , Q_2, \ldots, Q_{q})=
\{ (a_1, \dots, a_q) \in \mathbb{N}^q \mid \exists \ f \in\mathcal{F}_{q, r}, \,  (f)_\infty=a_1P_\infty+a_2Q_2+ \dots +a_q Q_{q}\},
$$
given in \cite{CT2}. The constructed sequence generalizes the results in  $\cite[\mbox{Theorems } 1, 2]{GOR}$ and $\cite[\mbox{Theorem } 4]{NX}$, see Remarks \ref{remark 1} and \ref{remark 2}. 

The paper is organized as follows. In Section II, we introduce two concepts of nonlinear complexity, and the function field $\mathcal{F}_{q,r}$. From a characterization of the Weierstrass semigroup at $q$ collinear rational places $P_\infty, Q_2, \ldots, Q_{q}$ on the function field $\mathcal{F}_{q,r}$, in Section III we present the new sequence over $\mathbb{F}_{q^{2r}}$.

\section{Preliminaries}

\subsection{Nonlinear complexity of sequences} 
The following nonlinear complexity notion was first defined by  Niederreiter and Xing in \cite{NX}, and it is also the main notion used in \cite{GOR} and \cite{YXY}. 
\begin{definition}\label{order nonlinear 1} Let $\mathbf{t}_n=(t_1, \ldots , t_n)$ be a sequence of length $n \geq 1$ over $\mathbb{F}_q$ and let $k \in \mathbb{N}$. If $t_i = 0$ for all $i=1, \ldots , n$, then we define the $k$-th \textit{order nonlinear complexity} $N^{(k)}(\mathbf{t}_n)$ to be $0$. Otherwise, $N^{(k)}(\mathbf{t}_1):=1$ and, for $n\geq2$, let $N^{(k)}(\mathbf{t}_n)$ be the smallest $m \in \mathbb{N}$ for which there exists a polynomial $f \in \mathbb{F}_q[x_1, \ldots , x_m]$, of degree at most $k$ in each variable, such that
	\begin{equation} \label{def nonlinear}
	t_{i+m} = f(t_i, t_{i+1}, \ldots , t_{i+m-1}), \mbox{ for } 1 \leq i \leq n-m.
	\end{equation}
\end{definition}

It is clear that $N^{(k)}(\mathbf{t}_n) \geq 0$ for all $n \geq 1$. If $n=2$, then $N^{(k)}(\mathbf{t}_2)=0$ or $1$. Now, for $n \geq 3$, since (\ref{def nonlinear}) is satisfied for $m=n-1$ with $f$ the constant polynomial $t_n$, we have that $0 \leq N^{(k)}(\mathbf{t}_n) \leq n-1$. Note that if $\mathbf{t}_n=(0, \ldots ,0, 1)$, then $N^{(k)}(\mathbf{t}_n)=n-1$.

Furthermore it suffices to consider $1 \leq k \leq q-1$, since any polynomial $f \in \mathbb{F}_q[x_1, \ldots , x_m]$ can be represented by a polynomial $g \in \mathbb{F}_q[x_1, \ldots , x_m]$ of degree at most $q-1$ in each variable, see \cite{LD}. For $k\geq q-1$, all nonlinear complexities $N^{(k)}(\mathbf{t}_n)$ of a fixed $\mathbf{t}_n$ are the same and equal to the called maximum-order complexity $M(\mathbf{t}_n) = N^{(q-1)}(\mathbf{t}_n)$ introduced by Jansen in \cite{J1} and \cite{J2}.

For a positive integer $n_1 \leq n$, a \textit{initial sequence} $\mathbf{t}_{n_1}$ of a given sequence $\mathbf{t}_n=(t_1, \ldots , t_n)$, means the sequence $\mathbf{t}_{n_1}=(t_1, \ldots , t_{n_1})$. 
\begin{lemma} \label{n1 n2}
	Let $\mathbf{t}_{n_1}$ and $\mathbf{t}_{n_2}$ be initial sequences of $\mathbf{t}_{n}$. If $n_1 \leq n_2$, then $N^{(k)}(\mathbf{t}_{n_1}) \leq N^{(k)}(\mathbf{t}_{n_2})$.
\end{lemma}
\begin{proof}
	Let $1 \leq n_1 \leq n_2 \leq n$ be integers, and consider the initial sequences $\mathbf{t}_{n_1}$ and $\mathbf{t}_{n_2}$ of the sequence $\mathbf{t}_n$. If $N^{(k)}(\mathbf{t}_{n_2})=m$, then there exists $f(x_1, \ldots , x_m)  \in \fq[x_1, \dots x_m]$ of degree at most $k$ in each variable satisfying 
	$$
	t_{i+m} = f(t_i, t_{i+1}, \ldots , t_{i+m-1}), \mbox{ for } 1 \leq i \leq n_2-m.
	$$
	Since $n_1 \leq n_2$, we have that the condition $t_{i+m} = f(t_i, t_{i+1}, \ldots , t_{i+m-1}), \mbox{ for } 1 \leq i \leq n_1-m$ 
	 is also satisfied. The result follows.
\end{proof}

A modified notion for nonlinear complexity, denoted by $L^{(k)}(\mathbf{t}_n)$, was also used in \cite{NX} and \cite{GOR}. The difference between the two notions is that the condition ``of degree at most $k$ in each variable" in $N^{(k)}(\mathbf{t}_n)$, is replaced with the condition ``of total degree at most $k$" in $L^{(k)}(\mathbf{t}_n)$.

\begin{definition}\label{order nonlinear 2} Given a sequence $\mathbf{t}_n=(t_1, \ldots , t_n)$ of length $n \geq 1$ over $\mathbb{F}_q$ and $k \in \mathbb{N}$, if $t_i = 0$ for all $i=1, \ldots , n$, we define the $k$-th \textit{order nonlinear complexity} $L^{(k)}(\mathbf{t}_n)$ to be $0$. Otherwise $L^{(k)}(\mathbf{t}_1):=1$ and, for $n\geq2$, let $L^{(k)}(\mathbf{t}_n)$ be the smallest $m \in \mathbb{N}$ for which there exists a polynomial $f \in \mathbb{F}_q[x_1, \ldots , x_m]$ of total degree at most $k$ such that
	$$
	t_{i+m} = f(t_i, t_{i+1}, \ldots , t_{i+m-1}), \mbox{ for } 1 \leq i \leq n-m.
	$$
\end{definition}

It is clear that $L^{(k)}(\mathbf{t}_n) \geq N^{(k)}(\mathbf{t}_n)$, for any $k$ 
and sequence $\mathbf{t}_n$. We also observe that $L^{(1)}(\bf{t_n})$ is not the same as 
the linear complexity $L(\bf{t_n})$ since for the computation of $L^{(1)}(\bf{t}_n)$ we accept degree one 
polynomials with constant term. In fact, we always have  $L(\mathbf{t}_n) \geq L^{(1)}(\mathbf{t}_n)$, and from a Remark on \cite[p. 401]{LD} we also have $L^{(1)}(\mathbf{t}_n)\geq L(\mathbf{t}_n)-1$.

Similarly to Lemma \ref{n1 n2}, we have that if $\mathbf{t}_{n_1}$ and $\mathbf{t}_{n_2}$ are initial sequences of $\mathbf{t}_{n}$ with $n_1 \leq n_2$, then $L^{(k)}(\mathbf{t}_{n_1}) \leq L^{(k)}(\mathbf{t}_{n_2})$.

The Hermitian function field $\cH$ is defined over $\mathbb{F}_{q^2}$ by the affine equation $y^q + y = x^{q+1}$. In \cite[Section IV]{NX}, Niederreiter and Xing constructed a sequence $\bold{s_M}$ from $\cH$ of length $\bold{M}=(q-1)(q^2-1)$. The obtained lower bounds on the complexities  from this sequence were closely connected to the existence of  a function $f$ in $\cH$  with pole divisor  supported in two rational places and of degree $q(q-1).$ The following result gives us the bounds obtained for $N^{(k)}(\mathbf{s}_n)$ and $L^{(k)} (\mathbf{s}_n)$ for initial sequences of $\bold{s_M}$.

\begin{theorem}\cite[Theorems 3 and 4]{NX}\label{teoNiede}
Let $\bold{s_n}$ be a initial sequence of $\bold{s_M}$. Then, for any integer $1 \leq k \leq q^2-1$, we have
\begin{equation} \label{cota N NX}
N^{(k)} (\mathbf{s}_n) \geq \dfrac{\lfloor \frac{n}{q^2-1} \rfloor (q^2-1) -1}{\lfloor \frac{n}{q^2-1} \rfloor + kq(q-1)},
\end{equation}
and
\begin{equation} \label{cota L NX}
L^{(k)} (\mathbf{s}_n) \geq \dfrac{\lfloor \frac{n}{q^2-1} \rfloor (q^2-1) -(q^2-q-1)k -1}{\lfloor \frac{n}{q^2-1} \rfloor + k},
\end{equation}
for the $k$-th order nonlinear complexity.
\end{theorem}

In \cite{GOR}, Geil, \"Ozbudak and Ruano constructed an analogous sequence to the one in the work of Niederreiter and Xing from the Hermitian function field $\cH$. With this new sequence, they improved the lower bounds in (\ref{cota N NX}) and (\ref{cota L NX}). Their improvements are based on the existence of a function $f$ in $\cH$ with poles exactly in  two rational places, and pole divisor of degree $2(q-1)$.

\begin{theorem}\cite[Theorems 1 and 2]{GOR} \label{teoFerruh}
Let $\bold{s_n}$ be a initial sequence of $\bold{s_M}$. Then, for any integer $1 \leq k \leq q^2-1$, we have
\begin{equation} \label{cota N GOR}
N^{(k)} (\mathbf{s}_n) \geq \dfrac{\lfloor \frac{n}{q^2-1} \rfloor (q^2-1) -(q -1)}{\lfloor \frac{n}{q^2-1} \rfloor + 2k(q-1)},
\end{equation}
and
\begin{equation} \label{cota L GOR}
L^{(k)} (\mathbf{s}_n) \geq \dfrac{\lfloor \frac{n}{q^2-1} \rfloor (q^2-1) -(k+1)(q -1)}{\lfloor \frac{n}{q^2-1} \rfloor + k(q-1)},
\end{equation}
for the $k$-th order nonlinear complexity.
\end{theorem}

\subsection{The maximal function field: $y^q + y = x^{q^r + 1} $} \label{kondo} 

Consider $\cf$ the function field over $\mathbb{F}_{q^{2r}}$ of  the curve defined  by the affine equation
   $$
y^q + y = x^{q^r + 1} ,
  $$
where $r \geq 1$ an odd integer. This function field $\cf$  has genus $g=q^{r}(q-1)/2$, and were first studied by Stichtenoth in \cite{Stich2}, and also by Garcia in \cite[Example 1]{Garcia}. 

For $r=1$, the function field $\mathcal{F}_{q, 1}$ is the Hermitian function field. So, in a way, $\cf$ can be seen as a generalization of the Hermitian function field.  

For each $x=a$ in $\mathbb{F}_{q^{2r}}$, the equation $y^q+y=a^{q^r+1}$ has $q$ distinct roots $y$ in $\mathbb{F}_{q^{2r}}$. And this determine $q^{2r+1}$ rational places on $\cf$, denoted by $P_{ab}$ where $ b^q+b=a^{q^r+1}, a, b $ in $\mathbb{F}_{q^{2r}}$. In addition to these places, there is a single place at infinity, $P_{\infty}$, the common pole of $x$ and $y$. So, $\cf$ has $q^{2r+1} + 1$ rational places and it is called a {\it maximal function field} over $\mathbb{F}_{q^{2r}}$ since its number of rational places equals the Hasse-Weil upper bound, namely equals $q^ {2r} + 1 + 2gq^r$.

Given a function $f$ in $\mathcal{F}_{q,r}$, we denote by $(f)$ and $(f)_\infty$ the divisor and pole divisor of $f$, respectively.

Fix $a \in \mathbb{F}_{q^{2r}}^{\ast}$ and let $P_{\infty}, Q_2=P_{ab_1}, Q_{3}= P_{ab_{2}},\ldots, Q_{q+1}=P_{ab_{q}}$ be $q+1$ pairwise distinct rational places in $\mathcal{F}_{q, r}$, where $b_\ell^q+b_\ell=a, b_\ell \in \mathbb{F}_{q^{2r}}$ for each $\ell=1, \dots, q$. Given an $m$-tuple ${\mathbf{s}}=(s_1,\ldots,s_m)$ in $\mathbb{N}^m$, and $i \in \mathbb{N}$, define the $m$-tuple
$$
\gamma_{\mathbf{s},i}:= (s_1 (q^r + 1)-iq, s_2(q^r + 1)+i, \ldots , s_m (q^r + 1) + i).
$$
By \cite[Theorem 3.4]{CT2}, for $2 \leq m \leq q+1$, the minimal generating set $\Gamma_m$ of the Weierstrass semigroup $H(P_{\infty},Q_2, \ldots, Q_{m}),$ is given by
\begin{equation} \label{Gamma}
\displaystyle \Gamma_m=\left\{  \gamma_{\mathbf{s},i} \mbox{; } \sum_{j=1}^{m}s_j=q+1-m \mbox{, } 0<iq<s_1(q^r + 1) \mbox{, } s_j\geq 0 \right\}.
\end{equation}
Therefore, for each $\gamma_{\mathbf{s},i} \in \Gamma_m$ there exists a rational function $f \in \mathcal{F}_{q, r}$ such that
\begin{equation} \label{funcao geral}
(f)_{\infty} = (s_1 (q^r + 1)-iq)P_{\infty} + (s_2(q^r + 1) + i)Q_2 + \cdots + (s_m(q^r + 1) + i)Q_{m}.
\end{equation}
By Equation (\ref{Gamma}), the degree of the pole divisor of $f$ is
$$
\displaystyle \deg((f)_{\infty}) = \sum_{j=1}^{m} s_j (q^r+1) - i (q-m+1) =(q-m+1)(q^r+1-i).
$$ 
Since $\Gamma_m$ is the minimal generating set of $H(P_{\infty}, Q_2, \ldots , Q_{m})$, it does not exist a rational function $g \in \mathcal{F}_{q, r}$ with 
$(g)_{\infty} =a_1P_\infty +a_2Q_2+ \dots +a_m Q_{m}$ and $\deg ((g)_{\infty}) < q^r-q^{r-1} +1$. So, we can
conclude that the smallest positive value for $\deg((f)_{\infty})$ is $q^r-q^{r-1}+1$, for $m=q$ and $i=q^{r-1}$.  We recall the improvements obtained in \cite{GOR}, of the results in \cite{NX}, were based on the existence of a function in $\mathcal{F}_{q, 1}$ with pole divisor in two places, and of small degree equal to $2(q-1)$.

By the definition of the set $\Gamma_q$, it exists a rational function $\varphi$ in $\mathcal{F}_{q, r}$ with pole divisor
\begin{equation} \label{funcao varphi}
(\varphi)_{\infty} =  P_{\infty} + q^{r-1}Q_2 + \cdots + q^{r-1}Q_q .
\end{equation}
This function $\varphi$ will play a fundamental role on the construction of the nonlinear sequence in the next section. 

In \cite{kondo}, the authors proved that the automorphism group of $\cf$ has a subgroup of order $q^{2r+1}(q^r+1)(q-1)$, consisting of all automorphism $\sigma$ defined by means of
\begin{equation} \label{automorfismo sigma}
\begin{array}{lll}
	\sigma: & x  \mapsto & \epsilon x + \xi, \\
	\\
	& y \mapsto & \epsilon^{q^r +1}y + \sum_{i=0}^{r-1} (-1)^i (\epsilon \xi^{q^r} x )^{q^{r-i-1}} + \beta,
\end{array}
\end{equation}
with $\xi, \beta, \epsilon$ in $ \mathbb{F}_{q^{2r}}$ satisfying $\beta^q + \beta = \xi^{q^r + 1}$, and $\epsilon^{(q^r+1)(q-1)}=1$. 

In the following, we collect some properties on automorphisms of function fields, that can be found in \cite[Section 8.2]{Stich} and will be useful for the results in the next section.
 \begin{lemma} \label{lemma auto}
Let $P$ be a rational place on $\mathcal{F}_{q, r}$, and $v_P$ be the discrete valuation at $P$. For any $\sigma \in \mbox{Aut}(\mathcal{F}_{q, r}/\mathbb{F}_{q^{2r}})$, and $f$ a function on the function field $\mathcal{F}_{q, r}$, we have
\begin{enumerate}[i)]
\item $\sigma(P)$ is also a rational place,
\item $v_{\sigma(P)} (\sigma(f)) = v_P (f)$, and
\item $\sigma(f) (\sigma(P)) = f(P)$ if $v_P(f) \geq 0$.
\end{enumerate}
\end{lemma}

\section{Sequences with high nonlinear complexity from $\cf$}

Let $\mathcal{F}_{q, r}$ and $P_{\infty}, Q_2=P_{ab_1}, Q_{3}=P_{ab_2}, \ldots, Q_{q}=P_{ab_{q-1}}$, with $a \in \mathbb{F}_{q^{2r}}^{\ast}$ and $b_i^q+b_i=a^{q^r+1}$, be as in Section \ref{kondo}. Consider the automorphism $\sigma$ in $Aut(\cf)$ given by
\begin{equation} \label{autom sigma}
\sigma(x) = \epsilon x,  \mbox{ } \mbox{ } \sigma(y) = \epsilon^{q^r+1} y, 
\end{equation}
where $\epsilon$ in $ \mathbb{F}_{q^{2r}}^{\ast}$ is an element of order  $(q^r+1)(q-1)$. Then $\sigma$ has order $(q^{r} + 1)(q-1)$, and fixes the places $P_{\infty}$ and $P_{0 0}$. The action of $\sigma$ on  the other $q^{2r+1}-1$ rational places of $\mathcal{F}_{q, r}$ provides a partition in $q(q^{r-1}+\cdots+q+1)+1$ orbits, where $q(q^{r-1}+\cdots+q+1)$ of them have length $(q^r+1)(q-1)$, and one has length $q-1$.

To simplify the notation we let $\mu:=(q^{r} + 1)(q-1)$ and $\nu:=q(q^{r-1}+\cdots+q+1)$. So, under the action of $\sigma$ on $Q_2, \ldots , Q_{q}$ we obtain $q-1$ orbits, each containing exactly $\mu$ distinct rational places of $\mathcal{F}_{q, r}$. Take $Q_{q+1},\ldots, Q_{\nu+1}$  other rational places on $\mathcal{F}_{q, r}$ such that, under the action of $\sigma$, we get $\nu-(q-1)$ distinct orbits of length $\mu$ pairwise distinct. Then we have the following $\nu$ distinct orbits
\begin{align*}
& Q_2, \sigma(Q_2), \dots , \sigma^{\mu -1}(Q_2),\\
& \vdots \\
& Q_{q}, \sigma(Q_{q}), \dots , \sigma^{\mu -1}(Q_{q})\\
& Q_{q+1}, \sigma(Q_{q+1}), \dots , \sigma^{\mu -1}(Q_{q+1}),\\
& \vdots \\
& Q_{\nu+1}, \sigma(Q_{\nu+1}), \dots , \sigma^{\mu -1}(Q_{\nu+1}).
\end{align*}
Let $\varphi$ be the rational function given in Equation (\ref{funcao varphi}). Since $\varphi$ has poles only in $P_\infty, Q_2, \dots , Q_{q}$, the sequence 
\begin{equation} \label{sequencia}
s_{(i-1)(\mu - 1) + j} : = \varphi(\sigma^{j} (Q_{i+1})), \text{ for } 1 \leq i \leq \nu \text{ and } 1 \leq j \leq \mu -1,
\end{equation}
is well defined. This sequence is denoted by $\mathbf{s}_M:=(s_1, \ldots , s_M)$, and has length $M=\nu (\mu - 1)=q(q^{r-1}+\cdots+q+1)(q^{r+1}-q^r+q-2)$.
Note that the sequence $\mathbf{s}_M$ is not identically zero since $\varphi$ has at most $q^r-q^{r-1}+1$ zeros. Explicitly, the sequence $\mathbf{s}_M$ is
$$
\begin{array}{llll}
s_1=\varphi(\sigma(Q_2))), & s_2 = \varphi(\sigma^2(Q_2)), & \dots & s_{\mu - 1} = \varphi(\sigma^{\mu - 1}(Q_2))\\
\vdots \\
s_{(q-2)(\mu - 1)+1}=\varphi(\sigma(Q_q)), & s_{(q-2)(\mu - 1)+2} = \varphi(\sigma^2(Q_q)), & \dots & s_{(q-1)(\mu - 1)} = \varphi(\sigma^{\mu - 1}(Q_q))\\
s_{(q-1)(\mu - 1)+1}=\varphi(\sigma(Q_{q+1})), & s_{(q-1)(\mu - 1)+2} = \varphi(\sigma^2(Q_{q+1})), & \dots & s_{q(\mu - 1)} = \varphi(\sigma^{\mu - 1}(Q_{q+1}))\\
\vdots & &  & \\
s_{(\nu -1)(\mu - 1)+1}=\varphi(\sigma(Q_{\nu+1})), &s_{(\nu -1)(\mu - 1)+2} = \varphi(\sigma^2(Q_{\nu+1})), & \dots & s_{\nu(\mu - 1)} = \varphi(\sigma^{\mu - 1}(Q_{\nu+1})).
\end{array}
$$

\begin{theorem} \label{teo 1}
Under notation and assumptions as above, let $1 \leq n \leq M=\nu (\mu - 1)$ be an integer and consider the initial sequence $\mathbf{s}_n = (s_1, \ldots , s_n)$ of the sequence $\mathbf{s}_M$ as above. For any integer $1 \leq k \leq \mu - 1$, we have

%\begin{equation}\label{cota N}
%N^{(k)} (\mathbf{s}_n) \geq \dfrac{\lfloor \frac{n}{\mu - 1} \rfloor (\mu - 1) -q^{r-1}(q -1)}{\lfloor \frac{n}{\mu - 1} \rfloor + k(q^r-q^{r-1}+1)},
%\end{equation}

\begin{equation}\label{cota N}
N^{(k)} (\mathbf{s}_n) \geq \dfrac{\lfloor \frac{n}{q^{r+1}-q^r+q-2} \rfloor (q^{r+1}-q^r+q-2) -q^{r-1}(q -1)}{\lfloor \frac{n}{q^{r+1}-q^r+q-2} \rfloor + k(q^r-q^{r-1}+1)},
\end{equation}
for the $k$-th order nonlinear complexity.

\end{theorem}
\begin{proof}
If $n < \mu -1=q^{r+1}-q^r+q-2$, then the lower bound is trivial. Hence we assume that $n \geq \mu - 1$ without loss of generality. Suppose that there exist a polynomial $f(x_1, \ldots, x_m) \in \mathbb{F}_{q^{2r}} [x_1, \dots x_m]$, where $1\leq m \leq \mu-2$, of degree at most $k$ in each variable such that
\begin{equation} \label{seq sm}
s_{m+\lambda} = f(s_{\lambda}, s_{\lambda+1}, \ldots , s_{\lambda+m-1}) \text{ for } 1 \leq \lambda \leq n-m,
\end{equation}
and $n \geq m+1$. By Lemma \ref{n1 n2}, it is enough to prove the result for $n = \delta (\mu-1)$, with $1 \leq \delta \leq \nu$. For $t=1, \ldots , \mu-1 - m$ and $i = 1,\ldots , \delta$ we have
$$
s_{(i-1)(\mu - 1) + m+t} = f(s_{(i -1) (\mu - 1) + t}, s_{(i-1) (\mu - 1) + t+1}, \ldots , s_{(i-1) (\mu - 1) + t+m-1}), 
$$
which is equivalent to 
\begin{equation} \label{eq hmt}
 \varphi(\sigma^{m+t}(Q_{i+1})) = f(\varphi(\sigma^t(Q_{i+1})), \ldots , \varphi(\sigma^{t+m-1}(Q_{i+1}))).
 \end{equation}
For $\varphi$ as in Equation (\ref{funcao varphi}), define the function $w$ in $\mathcal{F}_{q, r}$ as
$$w:= - \sigma^{-m}(\varphi) + f(\varphi,\sigma^{-1}(\varphi) , \ldots , \sigma^{-(m-1)}(\varphi)).$$
From Equation (\ref{funcao varphi}) we have $\mbox{deg}(\varphi)_{\infty} = q^r - q^{r-1} + 1, v_{P_{\infty}}(\varphi) = -1$, and $v_{Q_{k}}(\varphi) = -q^{r-1}$, for all $k=2,\ldots, q$. Thus, by Lemma \ref{lemma auto} (ii), it follows that $v_{\sigma^{-m}(Q_2)}(\sigma^{-m}(\varphi))=v_{Q_{2}}(\varphi)=-q^{r-1}$, and so the place $\sigma^{-m}(Q_2)$ is a pole of $\sigma^{-m}(\varphi)$. On the other hand, for each $b=0,\ldots , m-1$ we have $v_{\sigma^{-m}(Q_2)}(\sigma^{-b}(\varphi))=v_{Q_{2}}(\sigma^{m-b}(\varphi))$, and the divisor
$$(\sigma^{m-b}(\varphi))_{\infty} = P_{\infty} + q^{r-1} \sigma^{m-b}(Q_2) + \cdots + q^{r-1} \sigma^{m-b}(Q_{q}).$$
Then we conclude that $\sigma^{-m}(Q_2)$ is not a pole of  $\sigma^{-b}(\varphi).$ 
This ensures the function $w \neq 0$.

Now we are going to determine bounds on the number of poles and zeros of the function $w$.
%By Lemma \ref{lemma auto} (iii), if $v_Q(\varphi) \geq 0$, then, for an integer $j$, we have that $\varphi(\sigma^j(Q)) = \sigma^{-j} (\varphi(Q))$. 
The poles of $\varphi$ are exactly $P_{\infty}, Q_2, \ldots , Q_{q}$ and, by definition of $\sigma$, it follows that $v_{\sigma^k (Q_{i+1})}(\varphi) \geq 0$, for all $1 \leq k \leq \mu -1$ and $1 \leq i \leq \delta$. So, by Lemma \ref{lemma auto} (iii), it follows that for any integer $j$ 
$$
\varphi (\sigma^{j+k} (Q_{i+1})) = \varphi (\sigma^j (\sigma^k(Q_{i+1}))) = \sigma^{-j}(\varphi)(\sigma^k(Q_{i+1})).
$$
Thus, for $1 \leq t \leq \mu -1 - m$  we have
$$ 
\begin{array}{ll}
w(\sigma^{t}(Q_{i+1}))&=- \sigma^{-m}(\varphi)(\sigma^{t}(Q_{i+1}))  \\
& \\
&  + f(\varphi(\sigma^{t}(Q_{i+1})),\sigma^{-1}(\varphi)(\sigma^{t}(Q_{i+1})) , \ldots , \sigma^{-(m-1)}(\varphi)(\sigma^{t}(Q_{i+1})))\\
& \\
 & = -\varphi (\sigma^{m+t} (Q_{i+1})) + f(\varphi (\sigma^{t} (Q_{i+1})), \varphi (\sigma^{t+1} (Q_{i+1})), \ldots , \varphi (\sigma^{t+m-1} (Q_{i+1})))\\
 & \\
 & =0 \mbox{ (by (\ref{eq hmt}))}.
\end{array}
$$
Therefore we can conclude that
$$ 
 w(\sigma(Q_{i+1})) = \cdots = w(\sigma^{\mu -1 - m}(Q_{i+1}))=0, \mbox{ for all } i=1,\ldots,\delta.
$$
So, the zero divisor of the function $w$ has degree at least 
\begin{equation} \label{w0}
\mbox{deg}(w)_0 \geq \delta(\mu -1-m).
\end{equation}
The pole divisor of $\varphi$ is $(\varphi)_{\infty} = P_{\infty} +q^{r-1} Q_2 + \cdots + q^{r-1}Q_{q}$ 
and, by Lemma \ref{lemma auto} (ii)
$$
v_{\sigma^{-j} (P_{\infty})}(\sigma^{-j}(\varphi)) = v_{P_{\infty}} (\varphi) = -1 \text{ and } v_{\sigma^{-j} (Q_{k})}(\sigma^{-j}(\varphi)) = v_{Q_{k}} (\varphi) = -q^{r-1}, 
$$
for each $k=2,\ldots,q$ and  $j=0, \dots , m$. This yields
$$(\sigma^{-j}(\varphi))_{\infty} = P_{\infty} +q^{r-1} \sigma^{-j} (Q_2) + \cdots + q^{r-1}\sigma^{-j} (Q_{q}), \, j=0, \ldots, m.  $$
Thus
$$
\begin{array}{ll}
& (f(\varphi,\sigma^{-1}(\varphi) , \ldots , \sigma^{-(m-1)}(\varphi)))_{\infty} \leq \\ 
&\\
& km P_{\infty} + k \sum_{i=1}^{q-1} q^{r-1}Q_{i+1}+ k \sum_{i=1}^{q-1} q^{r-1}\sigma^{-1}(Q_{i+1}) + \cdots +  k \sum_{i=1}^{q-1} q^{r-1}\sigma^{-(m-1)}(Q_{i+1}).
 \end{array}
$$
%\lu{delete these:
%and
%$$
%(\sigma^{-m}(h_{\ell}))_{\infty} = P_{\infty} + q^{r-1}\sigma^{-m}(Q_2) + \ldots + q^{r-1}\sigma^{-m}(Q_{q}).
%$$
%}
Therefore,
$$
\mbox{deg}(w)_{\infty} \leq km  + kmq^{r-1}(q - 1) + q^{r-1}(q - 1) = km(q^r+1-q^{r-1}) + q^{r-1}(q-1).
$$
Since $\mbox{deg}(w)_{0} = \mbox{deg}(w)_{\infty}$, we can conclude that
$$
m \geq \dfrac{\delta(\mu - 1) - q^{r-1}(q - 1)}{k(q^r-q^{r-1}+1) + \delta}.
$$
\end{proof}

We observe that the longer the sequence $\mathbf{s}_n$, the greater the lower bound on $N^{(k)} (\mathbf{s}_n)$. For sequences over the same finite field $\mathbb{F}_{q^{2r}}$, with $r$ an odd integer, in Theorem \ref{teo 1} we can have higher values for the length of a sequence $\mathbf{s}_n$ compared to the ones given in Theorem \ref{teoFerruh}.

\begin{remark} \label{remark 1}
For $r=1$, we have a nonlinear sequence over the Hermitian function field, and if $q \geq 3$, the lower bound given in Theorem \ref{teo 1} is an improved lower bound on $N^{(k)} (\mathbf{s}_n)$, compared to the lower bound (\ref{cota N GOR}) in Theorem \ref{teoFerruh}. That is,
$$
\dfrac{\lfloor \frac{n}{q^2-2} \rfloor (q^2-2) -(q -1)}{\lfloor \frac{n}{q^2-2} \rfloor + kq} > \dfrac{\lfloor \frac{n}{q^{2}-1} \rfloor (q^{2}-1) -(q -1)}{\lfloor \frac{n}{q^{2}-1} \rfloor + 2k(q-1)}, \mbox{ for } q \geq 3.
$$

In fact, let $\eta_1:= \lfloor \frac{n}{q^2-1} \rfloor$ and $\eta_2:= \lfloor \frac{n}{q^2-2} \rfloor$. So $\eta_1 \leq \eta_2 \leq \eta_1 + 1$, that is, $\eta_2=\eta_1 + \eta$, where $\eta$ is equal to $0$ or $1$. We recall that $1 \leq \eta_1,\eta_2 \leq q$. Thus, we have that
$$
\dfrac{(\eta_1 + \eta) (q^2-2) -(q -1)}{(\eta_1 + \eta) + kq} > \dfrac{\eta_1 (q^{2}-1) -(q -1)}{\eta_1 + 2k(q-1)}
$$
if and only if
$$
 k \cdot \left[q^2 [(q-2)(\eta_1+\eta) + q\eta - 1] - q[3 \eta_1 + 4 \eta - 3] + 2[2(\eta_1 + \eta)-1] \right] +  \eta_1(\eta_1 + \eta) + \eta(q-1) >0.
$$
The last equality is true for $q \geq 3$.
\end{remark}

For the order nonlinear complexity $L^{(k)} (\mathbf{s}_n)$, given in Definition \ref{order nonlinear 2}, we have the following result.

\begin{theorem} \label{teo 2}
Under notation and assumptions as above, let $1 \leq n \leq M = \nu(\mu-1)$ be an integer and consider the initial sequence $\mathbf{s}_n = (s_1, \ldots , s_n)$ of the sequence $\mathbf{s_M}$ as above. For any integer $1 \leq k \leq \mu -1$, we have

\begin{equation}\label{cota N1}
L^{(k)} (\mathbf{s}_n) \geq \dfrac{\lfloor \frac{n}{q^{r+1}-q^r+q-2} \rfloor (q^{r+1}-q^r+q-2) -q^{r-1}(q -1) - k}{\lfloor \frac{n}{q^{r+1}-q^r+q-2} \rfloor + kq^{r-1}(q-1)}
\end{equation}

for the $k$-th order nonlinear complexity.
\end{theorem}
\begin{proof}
We follow the same approach as in the proof of Theorem \ref{teo 1},  but here we assume that $f(x_1, \ldots, x_m)$ is a polynomial in $\mathbb{F}_{q^{2r}}[x_1, \dots x_m]$ of total degree at most $k$. The proof is the same until Equation  (\ref{w0}), and the proof thereafter changes as follows.
Let $g(x_1, \ldots , x_m) = c x_1^{i_1} \ldots x_m^{i_m} \in \mathbb{F}_{q^{2r}}[x_1, \dots, x_m]$ be any monomial having a nonzero coefficient in $f(x_1, \ldots, x_m)$. We have $i_1 + \ldots + i_m \leq k$, and the pole divisor of the function 
$g(\varphi,\sigma^{-1}(\varphi) , \ldots , \sigma^{-(m-1)}(\varphi))$ satisfies 
\begin{align*}
(g(\varphi, \sigma^{-1}(\varphi) , \ldots , \sigma^{-(m-1)}(\varphi)))_{\infty} &\leq  k P_{\infty} + i_1 \sum_{i=1}^{q - 1} q^{r-1}Q_{i+1} \\
& +  i_2 \sum_{i=1}^{q - 1} q^{r-1}\sigma^{-1}(Q_{i+1}) + \ldots +  i_m \sum_{i=1}^{q - 1}q^{r-1} \sigma^{-(m-1)}(Q_{i+1}).
\end{align*}
For the pole divisor of $w$ we obtain
\begin{align*}
(w)_{\infty} &\leq k P_{\infty} + k \sum_{i=1}^{q - 1} q^{r-1}Q_{i+1} +  k\sum_{i=1}^{q - 1} q^{r-1}\sigma^{-1}(Q_{i+1}) + \ldots +  k \sum_{i=1}^{q - 1} q^{r-1}\sigma^{-(m-1)}(Q_{i+1})\\ 
& + q^{r-1}\sigma^{-m}(Q_2) + \ldots + q^{r-1} \sigma^{-m}(Q_q).
\end{align*}
In particular, we have
$$
\mbox{deg}(w)_{\infty} \leq k + km q^{r-1}(q-1) + q^{r-1}(q - 1).
$$
Since $\mbox{deg}(w)_0 \geq \delta(\mu -1-m)$ and $\mbox{deg}(w)_0 = \mbox{deg}(w)_{\infty}$, we conclude
$$
m \geq \dfrac{\delta(\mu-1) - q^{r-1}(q - 1) - k}{\delta + kq^{r-1}(q - 1)}.
$$
This completes the proof.
\end{proof}

\begin{remark} \label{remark 2}
The lower bound given in Theorem \ref{teo 2} is an improved lower bound on $L^{(k)} (\mathbf{s}_n)$, compared to the lower bound (\ref{cota L GOR}) in Theorem \ref{teoFerruh} when we look for the sequences in the same finite field $\mathbb{F}_{q^{2}}$, if $q\geq 5$.
\end{remark}

\noindent
{\bf Aknowledgments}

The authors thank the referees for the useful comments that have improved the presentation of this work. The first author thanks Fundação de Amparo à Pesquisa do Estado de Minas Gerais - FAPEMIG. The second author thanks Conselho Nacional de Desenvolvimento Científico e Tecnológico-CNPq, Fundação de Amparo à Pesquisa do Estado do Rio de Janeiro - FAPERJ, and Coordenação de Aperfeiçoamento de Pessoal de Nível Superior-CAPES for the partial support of this research. The third author thanks CNPq and FAPEMIG.

\bibliographystyle{spmpsci}

\end{document}